\documentclass[11pt]{amsart}
\usepackage{amscd}
\usepackage{mathrsfs}
\newtheorem{theorem}{Theorem}

\newtheorem{corollary}[theorem]{Corollary}

\newtheorem{definition}[theorem]{Definition}

\newtheorem{lemma}[theorem]{Lemma}
\newtheorem{notation}[theorem]{Notation}

\newtheorem{proposition}[theorem]{Proposition}
\newtheorem{remark}[theorem]{Remark}

\def\RE{\mathbb R}
\def\CO{{\mathbb C}}

\def\F{\mathcal F}
\def\S{\mathcal S}

\def\L{\mathscr L}

\newsymbol\subsetneq 2328

\begin{document}

\title[Metaplectic formulation...]{Metaplectic formulation of the Wigner transform and applications}

\author{Nuno Costa Dias \and Maurice A. de Gosson \and Jo\~{a}o Nuno Prata}

\email{ncdias@meo.pt}
\email{maurice.de.gosson@univie.ac.at}
\email{joao.prata@mail.telepac.pt }


\begin{abstract}
We show that the cross Wigner function can be written in the form $W(\psi, \phi)= \hat S (\psi \otimes
\overline{\hat\phi})$ where ${\hat\phi}$ is the Fourier transform of $\phi$ and $\hat S$ is a metaplectic
operator that projects onto a linear symplectomorphism $S$ consisting of a rotation along an ellipse in phase
space (or in the time-frequency space). This formulation can be extended to generic Weyl symbols and yields an
interesting fractional generalization of the Weyl-Wigner formalism. It also provides a suitable approach to
study the Bopp phase space representation of quantum mechanics, familiar from deformation quantization. Using
the "metaplectic formulation" of the Wigner transform we construct a complete set of intertwiners relating the
Weyl and the Bopp pseudo-differential operators. This is an important result that allows us to prove the
spectral and dynamical equivalence of the Schr\"odinger and the Bopp representations of quantum mechanics.
\end{abstract}
\maketitle


\section{Introduction}

It is well known that the cross Wigner function
\begin{equation}
W(\psi, \phi)(x,p)=\frac{1}{(2 \pi)^{n}} \int \, e^{ip\cdot\xi_p} \psi(x-\tfrac{\xi_p}{2})
\overline{\phi}(x+\tfrac{\xi_p}{2})\, d\xi_p \label{WT}
\end{equation}
defined for arbitrary $\phi,\psi \in L^2(\RE^n)$, satisfies the {\it Moyal identity}
$$
||W(\psi,\phi)||_{L^2(\RE^{2n})}=\frac{1}{(2 \pi)^{n/2}}||\phi ||_{L^2(\RE^{n})}||{\psi}||_{L^2(\RE^{n})} \, .
$$
Hence, for $||{\phi}||_{L^2(\RE^{n})}=1$, the mapping
\begin{equation}
W_\phi : L^2(\RE^n) \longrightarrow L^2(\RE^{2n});\, \psi \longmapsto W_\phi \psi ={(2 \pi)^{n/2}} W(\psi,\phi)
\label{WT2}
\end{equation}
called the {\it windowed Wigner transform}, is a non-surjective isometry $L^2(\RE^n) \to L^2(\RE^{2n})$.

In this paper we will further show that $W_\phi$ can be written in the form
\begin{equation}
W_\phi \psi = \hat S(\psi \otimes \overline{\hat \phi}) \label{MFWT}
\end{equation}
where $\hat S=e^{-i \theta_0 \hat H}$ is the unitary operator generated by the self-adjoint Hamiltonian
\begin{equation}
\hat H=2 \hat\xi_x\cdot\hat\xi_p-\left(\hat x\cdot\hat\xi_x \right)_+ -\left(\hat p\cdot\hat\xi_p\right)_++4\hat
x \cdot\hat p \,, \label{OpH}
\end{equation}
and $\theta_0$ is a suitable value of the scalar parameter. In (\ref{OpH}) the subscript $+$ denotes
symmetrization and
\begin{equation}
\hat x=x\cdot\quad , \quad \hat\xi_x=-i\partial_x \quad \mbox{and} \quad \hat p=p\cdot \quad , \quad
\hat\xi_p=-i\partial_p \label{SchRep}
\end{equation}
are the fundamental operators in the {\it phase space Schr\"odinger} representation of the Heisenberg algebra
(thus acting on functions $\Psi(x,p) \in D(\hat H) \subset L^2(\RE^{n}\times \RE^n)$) \cite{Dias12}.

Since $\hat H$ is quadratic, the transformation $\hat S$ is
metaplectic \cite{Folland,Gosson06,Gosson11}. Its projection onto
the group of linear symplectic transformations Sp$(4n)$ of
$\RE^{2n}\times \RE^{2n}$ is a symplectomorphism generated by the
classical Hamiltonian
\begin{equation}
H= 2\xi_x\cdot \xi_p -x\cdot\xi_x-p\cdot\xi_p+4x\cdot p \label{CH}
\end{equation}
and consists of a rotation along an ellipse in the phase space.

The "metaplectic formulation" of the Wigner function can be extended to generic Weyl operators. The resulting
formalism yields new calculation tools for the Weyl-Wigner calculus and displays several interesting
applications. We will use it to study a new topic and to revisit an old subject.

The new topic is a fractional generalization of the Weyl-Wigner formalism. Using the metaplectic approach we
construct a fractional version of both the Wigner function and the Weyl symbols. Moreover, we generalize several
of the key properties of the standard Weyl-Wigner framework to the fractional case. These include the regularity
properties, inversion and kernel formulas and the composition formula \cite{Gosson11,Wong}. We also study the
main features of the associated fractional quantization.

The old topic is the Bopp representation of quantum mechanics
\cite{bopp,Bracken,Dias12,Gosson11}. This is a phase space
operator representation that was recently used
\cite{Dias12,GoLu08} to prove precise spectral results for the
deformation quantization of Bayen {\it et al}
\cite{Bayen1,Bayen2,Dias04-1,Dias04-2}. A similar formulation was
used in \cite{Gosson08,GoLu10} to study the spectral properties of
generalized Landau operators \cite{LaLi97}, and in
\cite{DiGoLuPa10,DiGoLuPa11} to address the spectral problem in
noncommutative quantum mechanics. In another recent paper
\cite{Bondar}, a generalization of the Bopp representation was
used to address the problem of determining a consistent
formulation of coupled classical quantum dynamics
\cite{Dias01,Salcedo}. The results of \cite{Bondar} suggest that
the phase space representations are more suitable to address this
problem than the standard configurational space ones. Here, we
will focus on the mathematical formalism and use the metaplectic
representation of the Wigner transform to determine a complete
family of non-surjective isometries that intertwines the Weyl and
the Bopp pseudo-differential operators. A straightforward
consequence of this result is that the Schr\"odinger and the Bopp
representations of quantum mechanics (although not unitarily
related) display equivalent spectral and dynamical properties.

This paper is organized as follows: in the next section we present the metaplectic formulation of the Wigner
function (section 2.2) and Weyl symbols (section 2.3). In section 3 we define and study the fractional
generalization of the Weyl calculus. In section 4 we use the metaplectic formulation to study the spectral and
dynamical properties of the Bopp representation of quantum mechanics.

\begin{notation} A generic point in the double phase space $\RE^{4n}=\RE^{2n} \oplus
\RE^{2n}$ is denoted by $z=(x,p,\xi_x,\xi_p)$ where $(x,\xi_x)\in
\RE^{2n}$ and $(p,\xi_p)\in \RE^{2n}$ are canonical conjugate
pairs. The standard sympletic form on $\RE^{4n}$ is
$\sigma(z,z')=\xi_x \cdot x' +\xi_p \cdot p' -\xi_x' \cdot x
-\xi_p' \cdot p$ and the corresponding sympletic group is
Sp$(4n)$. The metaplectic group Mp$(4n)$ is the unitary
representation of the double cover of Sp$(4n)$.
\end{notation}

\begin{notation}
The norm and the inner product in $L^2(\RE^n)$ are denoted by $|| \cdot ||_{L^2(\RE^n)}$ and $( \cdot
,\cdot)_{L^2(\RE^n)}$, or just by $|| \cdot ||$ and $(\cdot , \cdot)$ whenever the dimension of the functions'
domain is clear from the context. The distributional bracket is $\langle \cdot , \cdot \rangle$.

For $f:\RE^n \longrightarrow \CO$ and $g:\RE^n \longrightarrow
\CO$, $f\otimes g$ denotes the tensor product function
$$
f\otimes g:\RE^n \times \RE^n \longrightarrow \CO; \, (x,y) \longmapsto f\otimes g (x,y)=f(x)g(y)
$$

Operators are denoted by roman letters with a hat (the exceptions
are $\hat\xi_x$ and $\hat\xi_p$). The hat also denotes the Fourier
transform (Notation 3) but it should always be clear from the
context what it refers to.
\end{notation}

\begin{notation}
The partial Fourier and inverse Fourier transforms are written:
$$
\hat\Psi(x,\xi_p)=\F_{p\rightarrow \xi_p}[\Psi(x,p)]=
\frac{1}{(2\pi)^{n/2}}\int_{\RE^n} \, e^{-i\xi_p\cdot p}\Psi(x,p)
\, dp
$$
and
$$
\check\Psi(x,p)=\F_{\xi_p\rightarrow p}^{-1}[\Psi(x,\xi_p)]=
\frac{1}{(2\pi)^{n/2}}\int_{\RE^n} \, e^{i\xi_p\cdot
p}\Psi(x,\xi_p) \, d\xi_p
$$
and are defined as unitary operators in $L^2(\RE^{2n})$ (for the Fourier transform in $L^2(\RE^{n})$ we use
exactly the same notation). As usual they can be extended by duality to $\S'(\RE^{2n})$. Notice that when we
write $\F_{p\rightarrow \xi_p}[\Psi(x,p)]$, $(x,p)$ is the argument of the original function and the result is a
function of $(x,\xi_p)$. On the other hand, if we write $\F_{p\rightarrow \xi_p}\Psi(x,\xi_p)$ then $(x,\xi_p)$
is already the argument of $\F_{p\rightarrow \xi_p}\Psi$.

\end{notation}

\section{Metaplectic formulation of the Wigner transform}

In this section we prove our main result. It states that $W(\psi,\phi)= (2\pi)^{-n/2} \hat S(\psi \otimes
\overline{\hat \phi})$ where $\hat S=e^{-i \theta_0 \hat H}$ is the unitary operator generated by the
Hamiltonian (\ref{OpH}) for $\theta_0=\tfrac{\sqrt{7}}{7}\arccos \tfrac{3}{4}$. We will also extend this
formulation to generic Weyl symbols.

\subsection{The classical symplectomorphism}

The unitary transformation $\hat S$ is generated by the quadratic
Hamiltonian (\ref{OpH}) and thus belongs to the metaplectic group
Mp$(\RE^{4n})$ \cite{Folland,Gosson97,Gosson11,Leray}. Its
projection onto the group of symplectic transformations of
$(\RE^{4n}, \sigma)$ is a symplectomorphism $S$ belonging to the
one-parameter group of symplectomorphisms $s(\theta)$ generated by
the classical Hamiltonian (\ref{CH}). In order to prepare our main
results, we now determine the explicit form of $s(\theta)$ and
$S$.

\begin{theorem}\label{CSym}
Let $s:\RE\times \RE^{4n} \longrightarrow \RE^{4n}$ be the one-parameter group of symplectic transformations
generated by the classical Hamiltonian
$$
H=2\xi_x\cdot \xi_p -x\cdot\xi_x-p\cdot\xi_p+4x\cdot p.
$$
For $\theta_0=\tfrac{\sqrt{7}}{7}\arccos \tfrac{3}{4}$ the symplectomorphism $S=s(\theta_0)$ is explicitly
\begin{equation}
S:\RE^{4n} \longrightarrow \RE^{4n}, \quad \left\{
\begin{array}{l}
x \longmapsto x/2+\xi_p/2 \\
\xi_x \longmapsto \xi_x -p\\
p \longmapsto p/2+\xi_x/2\\
\xi_p \longmapsto \xi_p -x
\end{array} \right. . \label{ST}
\end{equation}
Moreover, $S$ is the natural projection of $\hat S = e^{-i
\theta_0 \hat H}\in$Mp$(4n)$ onto Sp$(4n)$.
\end{theorem}

\begin{proof}
The map $s(\theta)=(x(\theta),p(\theta),\xi_x(\theta),\xi_p(\theta))$ is defined by the Hamilton equations
\begin{equation*}
\left\{ \begin{array}{l}
\dot x = \frac{\partial H}{\partial \xi_x}=  2\xi_p -x  \\
\dot p  = \frac{\partial H}{\partial \xi_p}= 2\xi_x -p \\
\dot \xi_x = -\frac{\partial H}{\partial x}= \xi_x -4p    \\
\dot \xi_p = -\frac{\partial H}{\partial p}= \xi_p -4x
\end{array} \right.
\end{equation*}
which decouple into two systems
\begin{equation*}
\left\{ \begin{array}{l}
\dot x = \frac{\partial H}{\partial \xi_x}=  2\xi_p -x \\
\dot \xi_p = -\frac{\partial H}{\partial p}=  \xi_p -4x
\end{array} \right. \quad ,\quad \left\{ \begin{array}{l}
\dot p  = \frac{\partial H}{\partial \xi_p}=  2\xi_x -p \\
\dot \xi_x = -\frac{\partial H}{\partial x}= \xi_x -4p
\end{array} \right.
\end{equation*}
with solutions ($k=\sqrt{7}$)
\begin{equation}
\left\{ \begin{array}{l} x(\theta)=x(0)\left[\cos k\theta - k^{-1} \sin k \theta \right] + 2k^{-1}\xi_p(0)
\sin k\theta \\
\\
\xi_p(\theta)=\xi_p(0)\left[\cos k\theta + k^{-1} \sin k\theta \right] -  4k^{-1} x(0) \sin k\theta
\end{array}
\right. \label{T1}
\end{equation}
and
\begin{equation}
\left\{ \begin{array}{l} p(\theta)=p(0)\left[\cos k\theta - k^{-1} \sin k\theta \right] +
2k^{-1}\xi_x(0) \sin k\theta \\
\\
\xi_x(\theta)=\xi_x(0)\left[\cos k\theta + k^{-1} \sin k\theta \right] - 4 k^{-1} p(0) \sin k \theta
\end{array}
\right. .\label{T2}
\end{equation}
Hence, $s(\theta)$ is given explicitly by the equations
(\ref{T1},\ref{T2}). For
$$
\theta=k^{-1}\arccos \tfrac{3}{4}=k^{-1} \arcsin \tfrac{k}{4}
$$
these equations yield the transformation $S$
exactly.

Since $H$ is the Weyl symbol of $\hat H$ (given by (\ref{OpH})),
the symplectomorphisms $s(\theta)$ are the projections onto
Sp$(4n)$ of the metaplectic operators $\hat U(\theta)=e^{-i\theta
\hat H}$. So, in particular, $S=s(\theta_0)$ is the projection of
$\hat S=\hat U(\theta_0)$.
\end{proof}

\begin{remark}
The family of transformations $s(\theta)$ is periodic. We have $s(\theta + \frac{2\pi}{k})=s(\theta)$ for all
$\theta \in \RE$. The orbits generated by $s$ belong to the following level surfaces
\begin{equation}
\left\{ \begin{array}{l}
2x^2+\xi_p^2-x\xi_p=Const\\
2p^2+\xi_x^2-p\xi_x=Const'
\end{array} \right.
\end{equation}
and consist of a rotation along two ellipses in the $x,\xi_p$ and $p,\xi_x$ planes of the phase space.
\end{remark}

\subsection{The main result}

Let $\hat T(\theta)$ be the one-parameter (periodic) group of unitary transformations defined by
\begin{equation}
\hat T(\theta): L^2(\RE^{2n})\longrightarrow L^2(\RE^{2n}); \, \Phi(x,\xi_p) \longmapsto \hat T(\theta)
\Phi(x,\xi_p)= \Phi(x(-\theta),\xi_p(-\theta)) \label{TT}
\end{equation}
where $x(-\theta)$ and $\xi_p(-\theta)$ are given by (\ref{T1}). Let us also define $\hat T=\hat T(\theta_0)$.

We start by proving the following

\begin{lemma}
The one-parameter unitary evolution group $\hat U(\theta)=e^{-i\theta \hat H}$ is given explicitly by
$$
\hat U(\theta): L^2(\RE^{2n}) \longrightarrow L^2(\RE^{2n})
$$
\begin{equation}
\hat U(\theta) \Psi(x,p) =\F_{\xi_p\rightarrow p}^{-1} \hat T(\theta) \F_{p\rightarrow \xi_p}\left[\Psi(x,
p)\right] \label{U}
\end{equation}

\end{lemma}

\begin{proof}

Since $\hat H$ is self-adjoint and $\S(\RE^{2n}) \subset D(\hat
H)$, the initial value problem
\begin{equation}
i \frac{\partial \Psi}{\partial \theta}= \hat H \Psi \quad, \quad
\Psi(\cdot,0)=\Psi_0(\cdot) \in \S(\RE^{2n}) \label{Scheq}
\end{equation}
has the unique solution
$$
\Psi(x,p,\theta) = e^{-i\theta \hat H} \Psi_0(x,p)
$$
where $e^{-i\theta \hat H}:L^2(\RE^{2n})\longrightarrow L^2(\RE^{2n})$ is the strongly continuous unitary
evolution group generated by $\hat H$ (see, for instance [Chapter 5,\cite{Oliveira}]). Our task is then to
determine the solution of (\ref{Scheq}) explicitly.

The Hamiltonian operator $\hat H$ (\ref{OpH}) can be re-written as
\begin{equation}
\hat H= \F_{\xi_p\rightarrow p}^{-1} \left[ -2i \xi_p \cdot
\partial_x +i x \cdot \partial_x-i \xi_p
\cdot\partial_{\xi_p}+4ix\cdot
\partial_{\xi_p} \right]\F_{p\rightarrow \xi_p}.
\label{FTOpH}
\end{equation}
where the symmetric terms of (\ref{OpH}) were calculated explicitly. Defining
$\hat\Psi(x,\xi_p,\theta)=\F_{p\rightarrow \xi_p}[\Psi(x,p,\theta)]$, we obtain from eq.(\ref{Scheq})
\begin{equation}
i \frac{\partial \hat\Psi}{\partial \theta}= \left[ -2i \xi_p\cdot
\partial_x +i x \cdot\partial_x-i \xi_p
\cdot\partial_{\xi_p}+4ix\cdot
\partial_{\xi_p}
\right]\hat\Psi \label{Scheq2}
\end{equation}
$$
\hat\Psi(\cdot,0)=\hat\Psi_0(\cdot) \in \S(\RE^{2n}).
$$
The solution of this initial value problem is easily found to be:
$$
\hat\Psi(x,\xi_p,\theta)=\hat\Psi_0(x(-\theta),\xi_p(-\theta))=\hat T(\theta) \hat\Psi_0(x,\xi_p).
$$
Moreover, if $\hat\Psi_0 \in \S(\RE^{2n})$ then $\hat\Psi(x,\xi_p,\theta) \in \S(\RE^{2n})$ for all $\theta$.

Consequently, the solution of eq.(\ref{Scheq}) is
$$
\Psi(x,p,\theta) =\F_{\xi_p\rightarrow p}^{-1} \left[\hat\Psi(x(-\theta),\xi_p(-\theta)) \right]
=\F_{\xi_p\rightarrow p}^{-1} \hat T(\theta) \F_{p \rightarrow \xi_p} \left[\Psi_0(x,p)\right].
$$
Let us then define
$$
\hat U(\theta):L^2(\RE^{2n}) \to L^2(\RE^{2n});\quad \hat U(\theta) \Psi_0(x,p) :=\F_{\xi_p\rightarrow p}^{-1}
\hat T(\theta) \F_{p \rightarrow \xi_p} \left[\Psi_0(x,p)\right].
$$
We conclude that $\hat U(\theta)$ is a linear and unitary (so continuous) operator and satisfies
$$
\left.\hat U(\theta)\right|_{\S(\RE^{2n})}=\left.e^{-i\theta \hat
H}\right|_{\S(\RE^{2n})}\, .
$$
Since both operators are continuous (and $\S(\RE^{2n})$ is dense
in $L^2(\RE^{2n})$) we also have in $L^2(\RE^{2n})$,
$$
\hat U(\theta)=e^{-i\theta \hat H}
$$
which completes the proof.
\end{proof}

It follows that

\begin{theorem}
Let $\hat S=\hat U(\theta_0)$ and $\theta_0=\frac{\sqrt{7}}{7}\arccos \left(\frac{3}{4}\right)$. For $\psi,\phi
 \in L^2(\RE^{n})$ we have
\begin{equation}
W(\psi,\phi)=\frac{1}{(2\pi)^{n/2}}\hat S \, (\psi \otimes \overline{\hat\phi}). \label{MWF}
\end{equation}
\end{theorem}

\begin{proof}
We notice that
$$
\overline{\hat\phi}(p)=\overline{\F_{y \rightarrow p}[\phi(y)]}=\F^{-1}_{y \rightarrow p}[\overline{\phi}(y)]
$$
and so
$$
\psi\otimes \overline{\hat\phi}(x,p)=\F^{-1}_{y \rightarrow p}[\psi\otimes\overline{\phi}(x,y)].
$$
We also notice from (\ref{T1}) and (\ref{TT}) that
$$
\hat T\Phi(x,\xi_p)=\Phi(x(-\theta_0),\xi_p(-\theta_0))=\Phi(x-\xi_p/2,x+\xi_p/2).
$$
Hence, we get
$$
\hat S (\psi \otimes \overline{\hat\phi})(x,p)= \F_{\xi_p\rightarrow p}^{-1} \hat T \F_{p\rightarrow \xi_p}
\F_{y\rightarrow p}^{-1} [\psi\otimes\overline{\phi}(x,y)]= \F_{\xi_p\rightarrow p}^{-1} \left[\hat T \psi
\otimes \overline{\phi}(x,\xi_p)\right]
$$
$$
=\frac{1}{(2 \pi)^{n/2}} \int_{\RE^n}  e^{ip\cdot\xi_p} \psi(x-\xi_p/2) \overline{\phi}(x+\xi_p/2) d\xi_p
$$
which is precisely $(2\pi)^{n/2} W(\psi,\phi)$.
\end{proof}

\subsection{Metaplectic formulation of Weyl symbols}

Let $\L(\S(\RE^n),\S'(\RE^n))$ be the space of linear and continuous operators of the form $\S(\RE^n)
\longrightarrow \S'(\RE^{n})$. The Schwartz kernel theorem states that all operators $\hat a \in
\L(\S(\RE^n),\S'(\RE^n))$ admit a kernel representation
\begin{equation}
\hat a \psi (x) = \langle K_a(x,\cdot),\psi(\cdot)\rangle \label{SchwartzKernel}
\end{equation}
where $K_a \in \S'(\RE^{n}\times \RE^n)$. The Weyl symbol of $\hat a$ is defined by
\begin{equation}
a(x,p)= \int_{\RE^n} e^{-ip \cdot y} K_a(x+\tfrac{1}{2}y,x-\tfrac{1}{2}y)\, dy \label{WeylSymbol}
\end{equation}
where the Fourier transform is taken in the distributional sense. We will denote by $\hat a
\overset{\text{Weyl}}{\longleftrightarrow} a$ the one-to-one association between pseudo-differential operators
(i.e. linear and continuous operators from $\S(\RE^n)$ to $\S'(\RE^n)$) and Weyl symbols.

From (\ref{SchwartzKernel}) we immediately obtain for all $\phi \in \S(\RE^n)$
\begin{equation}
\langle \hat a \psi,\overline{\phi} \rangle=\langle K_a(x,y), \psi \otimes \overline\phi(y,x) \rangle \label{SK}
\end{equation}
and it is well known that this equation can be re-written in terms of the cross Wigner function and the Weyl
symbol of $\hat a$, yielding the relation \cite{Gosson11}
\begin{equation}
\langle \hat a \psi,\overline{\phi} \rangle=\langle a, W(\psi,\phi) \rangle \, . \label{WK}
\end{equation}

Since the Wigner function $W(\psi,\phi)$ is (proportional to) the
Weyl symbol of the operator with kernel $\psi \otimes
\overline\phi$, we may suspect that a metaplectic formulation is
also possible for generic Weyl symbols. This result is the content
of the next theorem.

\begin{theorem}
Let $\hat a:\S(\RE^n) \longrightarrow \S'(\RE^n)$ be a generic pseudo-differential operator, and let $K_a\in
\S'(\RE^n\times\RE^n)$ be the kernel of $\hat a$. Then the Weyl symbol $a$ of $\hat a$ is given by
\begin{equation}
a(x,p)=(2\pi)^{n/2} \hat S \F_{y\to p}^{-1}\left[K_a(x,y)\right] \label{WeylMetaplectic}
\end{equation}
where $\hat S=\hat U(\theta_0)$ is given by (\ref{U}).
\end{theorem}

\begin{proof}
We note that for arbitrary $K_a\in \S'(\RE^n \times \RE^n)$ we have
$$
\F_{p\to \xi_p} \F_{y\to p}^{-1}\left[K_a(x,y) \right]=K_a(x,\xi_p).
$$
We also note from (\ref{U}) that $\hat S$ can be trivially extended to $\S'(\RE^{2n})$. Hence, from (\ref{U})
\begin{align*}
\hat S \F_{y\to p}^{-1}\left[K_a(x,y) \right] &= \F_{\xi_p\to p}^{-1}\left[\hat T K_a(x,\xi_p) \right]
=\F_{\xi_p\to p}^{-1}\left[ K_a(x-\tfrac{\xi_p}{2},x+\tfrac{\xi_p}{2}) \right]\\
&=\left(\frac{1}{2\pi}\right)^{n/2} \int_{\RE^n} e^{-ip\cdot \xi_p} K_a(x+\tfrac{\xi_p}{2},x-\tfrac{\xi_p}{2})\,
d\xi_p \\
&= \left(\frac{1}{2\pi}\right)^{n/2} a(x,p)
\end{align*}
which concludes the proof.
\end{proof}

\section{A Fractional generalization of the Wigner function and Weyl symbols}

In this section we use the metaplectic formulation of the Wigner function and Weyl symbols to obtain a natural
fractional generalization for both these objects. We also extend the main properties of the Wigner transform and
Weyl symbols to the fractional case. These results end up yielding simple, alternative proofs for several known
results about the Weyl-Wigner formalism.

\subsection{Main definitions and regularity results}

A natural definition of the fractional cross Wigner function is
\begin{definition}
Let $\psi,\phi \in L^2(\RE^n)$ and let $\hat U(\theta)$ be the unitary operator given by (\ref{U}). The {\it
fractional cross Wigner function} $W^\theta(\psi,\phi)$ is defined by
\begin{equation}
W^\theta(\psi,\phi)(x,p):= (2\pi)^{-n/2} \hat U(\theta) \psi\otimes\overline{\hat \phi}(x,p) \label{FWF}
\end{equation}
for $\theta \in \RE$. \label{DefFWF}
\end{definition}

Since $\hat U(\theta)$ is periodic, the family of quasi-distributions $(W^\theta(\psi,\phi))_{\theta \in \RE}$
is also periodic. We have $W^\theta(\psi,\phi)=W^{\theta+ \frac{2\pi}{\sqrt{7}}}(\psi,\phi)$ for all $\theta \in
\RE$. Some important elements of this family are $W^0(\psi,\psi)= (2\pi)^{-n/2} \psi \otimes
\overline{\hat\psi}$, which is $e^{ip\cdot x}$ times the anti-standard or Kirkwood quasi-distribution
\cite{Lee}. We also have $W^{2 \theta_0}(\psi,\psi)= (2\pi)^{-n/2} \overline{\psi} \otimes {\hat\psi}$ (check
eq.(\ref{AdjWigner})) which is $e^{-ip\cdot x}$ times the standard-ordered quasi-distribution \cite{Lee}.
Finally, $W^{\theta_0}(\psi,\phi)=W(\psi,\phi)$ is the cross Wigner function.

The metaplectic transformation $\hat U(\theta) \in$ Mp$(4n)$ is a unitary operator $L^2(\RE^{2n})
\longrightarrow L^2(\RE^{2n})$ and a continuous mapping $\S(\RE^{2n}) \longrightarrow \S(\RE^{2n})$ that extends
by duality to a continuous mapping $\S'(\RE^{2n}) \longrightarrow \S'(\RE^{2n})$. Hence,

\begin{theorem}
The sesquilinear map $W^\theta$ is a continuous mapping of the
forms:
$$
W^\theta:\S(\RE^n) \times \S(\RE^n) \longrightarrow \S(\RE^{2n})
$$
$$
W^\theta:L^2(\RE^n) \times L^2(\RE^n) \longrightarrow L^2(\RE^{2n})
$$
and extends to a continuous mapping of the form
$$
W^\theta:\S'(\RE^n) \times \S'(\RE^n) \longrightarrow \S'(\RE^{2n}) \, .
$$
\end{theorem}

\begin{proof}
Let $M(\psi,\phi)=\psi \otimes \overline{\hat\phi}$. Since the
Fourier transform $\F_{y \rightarrow p}$ is a continuous mapping
$\S(\RE^n) \longrightarrow \S(\RE^n)$, $L^2(\RE^n) \longrightarrow
L^2(\RE^n)$ and extends by duality onto a continuous mapping
$\S'(\RE^n) \longrightarrow \S'(\RE^n)$, the sesquilinear map $M$
can also be realized as a continuous mapping of any of the forms
$\S(\RE^n) \times \S(\RE^n) \longrightarrow \S(\RE^{2n})$,
$L^2(\RE^n)\times L^2(\RE^n) \longrightarrow L^2(\RE^{2n})$ and
also $\S'(\RE^n) \times \S'(\RE^n) \longrightarrow \S'(\RE^{2n})$.
The result follows from $W^\theta=(2\pi)^{-n/2}\hat U(\theta)M$.
\end{proof}

We also have for all $\theta \in \RE$
$$
(W^\theta(\psi_1,\phi_1),W^\theta(\psi_2,\phi_2))_{L^2(\RE^{2n})}=\frac{1}{(2\pi)^{n}}(\psi_1,\psi_2)_{L^2(\RE^{n})}
(\phi_1,\phi_2)_{L^2(\RE^{n})}
$$
which is the fractional generalization of the Moyal identity.
Hence, for $\phi_1=\phi_2=\phi$ such that
$||{\phi}||_{L^2(\RE^{n})}=1$, the {\it fractional windowed Wigner
transform} $W^\theta_\phi \cdot=(2\pi)^{n/2} W^\theta(\cdot,\phi)$
is a non-surjective isometry $L^2(\RE^n) \longrightarrow
L^2(\RE^{2n})$.

Just as in the case of the cross Wigner function, we may define a
fractional generalization of Weyl symbols.

\begin{definition}
Let $\hat a:\S(\RE^n) \longrightarrow \S'(\RE^n)$ be an arbitrary
pseudo-differential operator and let $K_a(x,y) \in \S'(\RE^n
\times \RE^n)$ be its kernel. The $\theta$-{\it Weyl symbol} of
$\hat a$ is defined by
\begin{equation}
a^\theta(x,p):= (2\pi)^{n/2} \hat U(\theta) \F_{y\to p}^{-1} \left[K_a(x,y) \right]. \label{FWS}
\end{equation}
\label{DefFWS}
\end{definition}

As an example, let $a^{\theta}_{\hat\rho}$ be the $\theta$-Weyl
symbol of the rank one operator $\hat\rho$ defined for fixed
$\psi,\phi \in L^2(\RE^n)$ by
$$
\hat\rho:\S(\RE^n) \longrightarrow \S'(\RE^n),\, \xi \longmapsto \hat\rho \xi=(\xi,\phi)_{L^2(\RE^n)} \psi.
$$
We have, as in the standard Weyl-Wigner formalism, $a^{\theta}_{\hat\rho}=(2\pi)^{n}W^\theta(\psi,\phi)$.

\subsection{Inversion and Kernel formulas}
Let, as before,
\begin{equation}
W_\phi^\theta : L^2(\RE^n)\longrightarrow L^2(\RE^{2n}); \, \phi \longmapsto W_\phi^\theta \psi=(2\pi)^{n/2}
W^\theta(\psi,\phi) \label{TWT}
\end{equation}
be the $\theta$-windowed Wigner transform defined for $\theta \in \RE$ and some window $\phi \in \S(\RE^n)$ such
that $||\phi||=1$. For $\theta=\theta_0$ we write $W_\phi \equiv W_\phi^{\theta_0}$ as in eq.(\ref{WT2}).  Let
also $(W^\theta_\phi)^{\ast }:L^2(\RE^{2n})\longrightarrow L^2(\RE^n)$ be the mapping
\begin{equation}
(W^\theta_\phi)^{\ast }\Psi(\cdot):= \int_{\RE^n} \, \hat\phi(p)
\hat U^{-1}(\theta) \Psi(\cdot,p) \, dp \, . \label{FIWF}
\end{equation}
Then

\begin{theorem}
The mapping $(W^\theta_\phi)^{\ast }$ is the adjoint of
$W^\theta_\phi$ and satisfies:

(i) $(W^\theta_\phi)^{\ast } W_\phi^\theta=1$;

(ii) $W^\theta_\phi (W^\theta_\phi)^{\ast } =P^\theta_\phi$ where
$P^\theta_\phi$ is the orthogonal projection onto the range of
$W^\theta_\phi$;

(iii) $(W^\theta_\phi)^{\ast }$ extends to a continuous operator
$\S'(\RE^{2n}) \longrightarrow \S'(\RE^n)$ defined by
\begin{equation}
\langle (W^\theta_\phi)^{\ast } \Psi , \xi \rangle = \langle \Psi, \hat U^{-1}(\theta) \xi\otimes \hat\phi
\rangle,\quad \forall \xi \in \S(\RE^n). \label{FIWF2}
\end{equation}
\label{Theorem12}
\end{theorem}

\begin{proof}
The adjoint of $W^\theta_\phi$ is defined by
$$
((W^\theta_\phi)^{*} \Psi,\xi)=(\Psi,W^\theta_\phi\xi) , \quad \forall \xi \in L^2(\RE^n)
$$
and since
\begin{align*}
\left(\Psi,W^\theta_\phi\xi\right) &=\left(\Psi,\hat U(\theta)\xi \otimes \overline{\hat\phi}\right)=
\left(\hat U^{-1}(\theta)\Psi,\xi \otimes \overline{\hat\phi}\right)\\
&= \int_{\RE^n} \,\overline{\xi}(x) \left[\int_{\RE^n} \,\hat\phi(p) \hat U^{-1}(\theta) \Psi(x,p) \, dp \right]
dx \, ,
\end{align*}
we have
$$
(W^\theta_\phi)^{*} \Psi=\int_{\RE^n} \,\hat\phi(p) \hat
U^{-1}(\theta) \Psi(x,p) \, dp
$$
as claimed.

(i) Let now $\Psi=W^\theta_\phi \psi= \hat U(\theta) (\psi \otimes \overline{\hat\phi})$. Then
\begin{align*}
(W^\theta_\phi)^{\ast } \Psi=&\int_{\RE^n} \,\hat\phi(p) \hat
U^{-1}(\theta) \hat U(\theta) (\psi \otimes
\overline{\hat\phi})(x,p) \, dp\\
=&\int_{\RE^n} \,\hat\phi(p) \overline{\hat\phi}(p) \psi(x) \, dp=\psi(x).
\end{align*}

(ii) Conversely, let $\Psi \in L^2(\RE^{2n})$ and let
$P^\theta_\phi=W^\theta_\phi (W^\theta_\phi)^{\ast }$. Then, of
course
$$
P^\theta_\phi P^\theta_\phi=W^\theta_\phi (W^\theta_\phi)^{\ast
}W^\theta_\phi (W^\theta_\phi)^{\ast }=P^\theta_\phi
$$
and if $\Psi \in $ Ran $W^\theta_\phi$ then $\Psi=W^\theta_\phi\psi$ for some $\psi\in L^2(\RE^n)$ and
$$
P^\theta_\phi\Psi=W^\theta_\phi (W^\theta_\phi)^{\ast
}W^\theta_\phi \psi=W^\theta_\phi \psi=\Psi.
$$
Hence, the claim.

(iii) As a distribution $(W^\theta_\phi)^{\ast }\Psi$ is
completely defined by its action on arbitrary test functions $\xi
\in \S(\RE^n)$. For $\Psi \in \S(\RE^{2n})$ and $\phi \in
\S(\RE^n)$, we have from (\ref{FIWF})
$$
\langle (W^\theta_\phi)^{\ast } \Psi, \xi\rangle=\langle \hat
U^{-1}(\theta) \Psi, \xi \otimes \hat\phi \rangle=\langle \Psi,
\hat U^{-1}(\theta) \xi\otimes \hat\phi \rangle \, .
$$
Since $\hat U^{-1}(\theta):\S'(\RE^{2n}) \longrightarrow
\S'(\RE^{2n})$, this formula extends trivially to $\Psi \in
\S'(\RE^{2n})$.
\end{proof}

All these considerations are of course valid for the standard windowed Wigner transform, which is just a
particular element of the family $W^\theta_\phi$, $\theta \in \RE$.

To proceed let us generalize the Kernel formula (\ref{WK}) for the fractional case. Let the kernel of $\hat
a:\S(\RE^n) \longrightarrow \S'(\RE^n)$ be given by $K_a(x,y) \in \S'(\RE^n \times \RE^n)$ so that
$$
\hat a\psi(x)=\langle K_a(x,y),\psi(y) \rangle.
$$
We then have
\begin{equation}
\langle \hat a \psi, \overline{\phi} \rangle =\langle K_a(x,y), \psi(y) \overline{\phi}(x) \rangle \label{KerF}
\end{equation}
and it is well known that this formula can be re-expressed in terms of the Wigner function as in (\ref{WK}).
This is because
\begin{equation}
\langle K_a(x,y), \psi(y) \overline{\phi}(x) \rangle = \langle a(x,p), W(\psi,\phi)(x,p) \rangle \,
.\label{KerW}
\end{equation}

Let also $\hat b=\hat a^{\dagger}$ be the formal adjoint of $\hat a$, defined by
\begin{equation}
\langle \hat b \psi, \overline{\phi} \rangle= \overline{ \langle \hat a \phi, \overline{\psi} \rangle} \quad ,
\quad \forall \psi,\phi \in \S(\RE^n) \, . \label{adjoint}
\end{equation}
It is well known that $K_b(x,y)=\overline{K_a}(y,x)$. This relation follows immediately from eq.(\ref{adjoint})
by taking into account that
$$
\langle \hat b \psi, \overline{\phi} \rangle=\langle K_b(x,y), \psi(y)  \overline{\phi}(x) \rangle
$$
and that
$$
\overline{\langle \hat a \phi, \overline{\psi} \rangle}= \langle \overline{K_a}(y,x), \psi(y) \overline{\phi}(x)
\rangle.
$$

We can now state the following

\begin{theorem}
Let $\hat a:\S(\RE^n) \longrightarrow \S'(\RE^n)$ be an arbitrary pseudo-differential operator and let $\hat b$
be its formal adjoint. Let $b^\theta$ be the $\theta$-Weyl symbol of $\hat b$. Then
\begin{equation}
\langle \hat a \psi, \overline{\phi} \rangle=\langle \overline{b^\theta},W^\theta(\psi,\phi) \rangle
\label{KerFF}
\end{equation}
for all $\psi,\phi \in \S(\RE^n)$.
\end{theorem}

\begin{proof}
We have from (\ref{KerF})
\begin{align*}
\langle \hat a \psi, \overline{\phi} \rangle & =\overline{ \langle \overline{K_a}(y,x),
\overline{\psi(x)  \overline{\phi}(y)} \rangle }\\
&=\overline{ \langle \hat U(\theta) \F_{y\to p}^{-1} \left[\overline{K_a}(y,x)\right], \overline{ \hat U(\theta)
\F_{y\to p}^{-1} \left[  \psi(x)  \overline{\phi}(y) \right]} \rangle } \\
&=(2\pi)^{n/2}\langle \overline{ \hat U(\theta) \F_{y\to p}^{-1} \left[\overline{K_a}(y,x)\right]},
W^\theta(\psi,\phi)(x,p) \rangle \, .
\end{align*}
To complete the proof we just notice that $K_b(x,y)=\overline{K_a}(y,x)$ and use the definition of $\theta$-Weyl
symbol (\ref{FWS}).
\end{proof}

Finally, notice that in the Weyl case $\theta=\theta_0$ we have $\overline{b^{\theta_0}}=a^{\theta_0}=a$ and so
formula (\ref{KerFF}) reduces to the standard formula (\ref{WK}).

\subsection{$\theta$-Weyl calculus and quantization}

From the definitions \ref{DefFWF} and \ref{DefFWS} we immediately realize that
\begin{eqnarray}
W^\theta(\psi,\phi) &=& \hat U(\theta-\theta_0) W(\psi,\phi) \nonumber \\
a^\theta &=& \hat U(\theta-\theta_0) a \, . \label{SymSym}
\end{eqnarray}
where $a \equiv a^{\theta_0}$ and $a^\theta$ are the Weyl and the $\theta$-Weyl symbols of an arbitrary
pseudo-differential operator $\hat a$, respectively. These formulas allow us to relate the standard and the
fractional Weyl calculus. Let us denote by $\hat a \overset{\theta}{\longleftrightarrow} a^\theta $ the
association between the $\theta$-symbol $a^\theta= \hat U(\theta-\theta_0) a$ and the pseudo-differential
operator $\hat a \overset{\text{Weyl}}{\longleftrightarrow} a$. Then

\begin{theorem}
Let $\hat a: \S(\RE^n) \longrightarrow \S'(\RE^n)$ and $\hat b:\S(\RE^n) \longrightarrow \S(\RE^n)$ be two
arbitrary pseudo-differential operators. Let $\hat a \overset{\theta}{\longleftrightarrow} a^\theta $ and $\hat
b \overset{\theta}{\longleftrightarrow} b^\theta $. Then
$$
\hat a \hat b \overset{\theta}{\longleftrightarrow} a^\theta *_\theta b^\theta
$$
where the product $*_\theta$ is given by
\begin{equation}
a^\theta *_\theta b^\theta = \hat U(\theta -\theta_0) \left[ \left(\hat U^{-1}(\theta -\theta_0) a^\theta
\right) *_M \left(\hat U^{-1}(\theta -\theta_0) b^\theta \right) \right] \, . \label{Tproduct}
\end{equation}
Here $*_M$ is the standard Moyal product \cite{Gosson11}
$$
a*_Mb(z)=\left( \tfrac{1}{4\pi} \right)^{2n} \int_{\RE^{4n}} e^{\tfrac{i}{2} \sigma(z',z'')} a(z+\tfrac{1}{2}z')
b(z-\tfrac{1}{2}z'')\, dz' \, dz'' \,
$$
where $z=(x,p)$ are canonical coordinates, and $\sigma$ is the standard symplectic form on $\RE^{2n}$.
\end{theorem}

\begin{proof}
Let $\hat a \overset{\rm Weyl}{\longleftrightarrow} a $ and $\hat b \overset{\rm Weyl}{\longleftrightarrow} b$.
Then \cite{Gosson11}
$$
\hat a \hat b \overset{\rm Weyl}{\longleftrightarrow} a *_M b \, .
$$
The association between an arbitrary $\hat a $ and its $\theta$-symbol $a^\theta$ is defined by $\hat a
\overset{\theta}{\longleftrightarrow} a^\theta = \hat U(\theta -\theta_0) a$ where $\hat a \overset{\rm
Weyl}{\longleftrightarrow} a $. Hence, the $\theta$-symbol of $\hat a \hat b$ is
$$
a^\theta *_\theta b^\theta = \hat U(\theta-\theta_0) (a *_M b) \, .
$$
Using $a= \hat U^{-1} (\theta-\theta_0) a^\theta $ and $b= \hat U^{-1} (\theta-\theta_0) b^\theta $ we get
(\ref{Tproduct}).

\end{proof}

It easily follows from (\ref{Tproduct}) that $*_\theta$ is non-commutative, distributive and associative. Let us
prove this last property. Since $*_M$ is associative we have (let $\hat U= \hat U(\theta-\theta_0)$)
\begin{eqnarray}
(a^\theta *_\theta b^\theta) *_\theta c^\theta & = & \hat U \left[ \left[ \left(\hat U^{-1} a^\theta \right) *_M
\left(\hat U^{-1} b^\theta \right) \right] *_M \left(\hat U^{-1} c^\theta \right) \right] \nonumber \\
&=& \hat U \left[  \left(\hat U^{-1} a^\theta \right) *_M \left[
\left(\hat U^{-1} b^\theta \right)  *_M \left(\hat U^{-1} c^\theta \right) \right]\right] \nonumber \\
&=& a^\theta *_\theta  (b^\theta *_\theta c^\theta ) \, .
\end{eqnarray}

To proceed let us determine the $\theta$-symbol of the adjoint operator. Let, as before, $\hat b=\hat
a^{\dagger}$ be the formal adjoint of $\hat a$. Then the Weyl symbols of $\hat a$ and $\hat b$ satisfy
$a=\overline{b}$ and we have for the $\theta$-symbols
\begin{equation}
b^{\theta_0+\alpha}=\hat U(\alpha) b= \hat U(\alpha) \overline{a}= \overline{\hat U^{-1}(\alpha)
a}=\overline{a^{\theta_0-\alpha}} \, . \label{AdjSym}
\end{equation}
Hence, and in general, the $\theta$-symbol of a self-adjoint operator is not real. The exception is the standard
Weyl case.

Using the result (\ref{AdjSym}) we can re-write the kernel formula (\ref{KerFF}) in the form
\begin{equation}
\langle \hat a \psi, \overline{\phi} \rangle=\langle a^{\theta_0-\alpha},W^{\theta_0+\alpha}(\psi,\phi) \rangle
\label{KerA}
\end{equation}
for all $\alpha \in \RE$.

Finally, let us remark on a few features of the fractional Weyl quantization. The next result is a trivial
consequence of (\ref{KerA})

\begin{corollary}
The expectation value of a quantum observable $\hat a$ in a state
$\psi \in \S(\RE^n)$ is given in terms of its $\theta$-symbol by
\begin{equation}
\langle \hat a \psi, \overline{\psi} \rangle= \int_{\RE^{2n}} \overline{a^\theta(z)} W^\theta(\psi,\psi)(z) \,
dz \, . \label{AvgT}
\end{equation}
where the integral is interpreted in the distributional sense.
\end{corollary}

\begin{proof}
The proof follows from eqs.(\ref{AdjSym}) and (\ref{KerA}) by noticing that for $\hat a$ self-adjoint we have
$\hat b=\hat a$ and so $a^{\theta_0-\alpha}=\overline{a^{\theta_0+\alpha}}$. Notice also that in the standard
Weyl case, we have in addition $a^{\theta_0} \equiv a=\overline{a}$.
\end{proof}

From eq.(\ref{KerA}) we can also determine the marginal probability distributions for the fractional Wigner
function. Let, for instance, $\hat a$ be the projection operator given formally by $|x_0><x_0|$ where $x_0 \in
\RE^n$. We have explicitly
$$
\hat a: \S(\RE^n) \longrightarrow \S'(\RE^n); \, \psi \longmapsto \psi(x_0) \delta(x-x_0)
$$
where $\delta$ is the Dirac delta distribution. The Weyl symbol of
$\hat a$ can be easily calculated from eq.(\ref{WeylSymbol}):
$a(x,p)=\delta(x-x_0)$. Hence, the marginal probability
distribution for the position observable is simply
\begin{eqnarray}
{\mathcal P}(x_0) = \langle \hat a \psi,\overline{\psi}\rangle & = & \int_{\RE^n}\int_{\RE^n} \delta(x-x_0)
W(\psi,\psi)(x,p) \, dx dp \nonumber \\
& = & \int_{\RE^n} W(\psi,\psi)(x_0,p) \, dp \, \nonumber.
\end{eqnarray}
In the fractional case, we also have from eq.(\ref{SymSym}) and eq.(\ref{AvgT})
$$
{\mathcal P}(x_0)=\int_{\RE^n}\int_{\RE^n} \left(\hat U^{-1}(\theta-\theta_0) \delta(x-x_0)\right)
W^\theta(\psi,\psi)(x,p) \, dx dp \, .
$$
However, in general, this formula does not further simplify as in the case of standard Wigner functions.

Another important feature of the Wigner function $W(\psi,\psi)$ is that it is a real quasidistribution. This
property is also not shared by the fractional Wigner functions. For the $\theta$-cross Wigner function we get
from (\ref{SymSym})
\begin{equation}
W^{\theta_0 + \alpha}(\psi,\phi)= \overline{W^{\theta_0-\alpha}(\phi,\psi)} \label{AdjWigner}
\end{equation}
which yields a reality condition only for the case $\alpha=0$ and $\phi=\psi$.

\section{The Bopp representation of quantum
mechanics}

In this section we show that the windowed Wigner transform $W_\phi$ intertwines the Schr\"odinger and the Bopp
representations of quantum mechanics and discuss the implications of this result for the spectral and dynamical
properties of operators in the two representations. This subject was previously studied in several papers
\cite{Dias12,GoLu08,GoLu10}. Here, we present an alternative approach using the metaplectic formalism developed
in the previous sections.

Let $\hat a:\S(\RE^{n}) \longrightarrow \S'(\RE^{n})$ be a generic pseudo-differential operator. In terms of its
Weyl symbol $a \in \S'(\RE^n \times \RE^n)$, the operator can be written \cite{Shubin}
\begin{equation}
\hat a \psi (x)= \frac{1}{(2\pi)^{n}} \int_{\RE^{2n}}
e^{i(x-x')\cdot \xi_x} a (\frac{1}{2}(x+ x'), \xi_x) \psi(x')\,
dx'd\xi_x \label{PDO}
\end{equation}
where the integral is well defined for $a \in \S(\RE^{n}\times \RE^n)$ and should otherwise be interpreted in
the distributional sense. The operator $\hat a$ is formally $\hat a=a(x,-i\partial_x)$ and the mapping $a
\longmapsto \hat a$ yields a precise definition of the {\it standard Schr\"odinger representation} of quantum mechanics.\\

A trivial extension of $\hat a$ to phase space functions is given by the operator $\hat A:\S(\RE^n \times \RE^n)
\longrightarrow \S'(\RE^n \times \RE^n)$ of the form:
\begin{equation}
\hat A \Psi (x,p)= \frac{1}{(2\pi)^{n}} \int_{\RE^{2n}}
e^{i(x-x')\cdot \xi_x} a (\frac{1}{2}(x+ x'), \xi_x) \Psi(x',p)\,
dx'd\xi_x .\label{PSPDO}
\end{equation}
Then

\begin{theorem}
The operator $\hat A$ is a pseudo-differential operator with Weyl symbol
\begin{equation}
A(x,p;\xi_x,\xi_p)=a(x,\xi_x). \label{Sym}
\end{equation}
Moreover, for an arbitrary wave function $\Psi(x,p)=\psi \otimes \overline{\hat\phi}(x,p)\in \S(\RE^{2n})$, we
have
\begin{equation}
\hat A (\psi \otimes \overline{\hat\phi})=(\hat a\psi) \otimes
\overline{\hat\phi}. \label{Int}
\end{equation}
\end{theorem}

\begin{proof}
Let $\hat A$ be the pseudo-differential operator with Weyl symbol $A$ of the form (\ref{Sym}). Then $\hat
A\Psi(x,p)$ is explicitly
\begin{eqnarray*}
&&\hat A \Psi (x,p)\\
&=& \frac{1}{(2\pi)^{2n}} \int_{\RE^{4n}} e^{i[(x-x')\cdot
\xi_x+(p-p')\cdot \xi_p]} A (\frac{x+ x'}{2},\frac{p+ p'}{2},
\xi_x, \xi_p) \Psi(x',p')\,
dx'dp'd\xi_xd\xi_p\\
&=&\frac{1}{(2\pi)^{2n}} \int_{\RE^{4n}} e^{i[(x-x')\cdot
\xi_x+(p-p')\cdot \xi_p]} a (\frac{1}{2}(x+ x'), \xi_x)
\Psi(x',p')\, dx'dp'd\xi_xd\xi_p\\
& =& \frac{1}{(2\pi)^{n}} \int_{\RE^{2n}} e^{i(x-x')\cdot \xi_x} a
(\frac{1}{2}(x+ x'), \xi_x) \Psi(x',p)\, dx'd\xi_x
\end{eqnarray*}
which coincides with the expression (\ref{PSPDO}) exactly.

The intertwining relation (\ref{Int}) follows directly from
(\ref{PSPDO}) for $\Psi(x,p)=\psi(x)\overline{\hat\phi}(p)$.
\end{proof}

Just like $\hat a$, the operator $\hat A$ is also formally $\hat A=a(x,-i\partial_x)$ (but now the fundamental
operators $x\cdot$ and $-i\partial_x$ act on phase space functions). The mapping $a \longmapsto \hat A$ yields
the {\it phase space Schr\"odinger representation} of quantum mechanics. The spectral and dynamical properties
of the operators $\hat a$ and $\hat A$ are equivalent. This is an important property that we now discuss in some
detail. For complete proofs the reader should refer to \cite{Dias12}.

For an arbitrary window $\phi \in \S(\RE^n)$ such that $||\phi||_{L^2(\RE^n)}=1$, let
$$
T_\phi:\S'(\RE^n) \longrightarrow \S'(\RE^{2n}); \psi \longmapsto \Psi= T_\phi \psi=\psi \otimes
\overline{\hat\phi} \, ,
$$
and let
$$
T^*_\phi:\S'(\RE^{2n}) \longrightarrow \S'(\RE^{n}); \Psi \longmapsto \psi= T^*_\phi \Psi
$$
be defined by
$$
\langle T^*_\phi \Psi, \xi \rangle = \langle \Psi, \xi  \otimes \hat\phi \rangle \quad , \quad \forall \xi \in
\S(\RE^n) \, .
$$
These two maps coincide exactly with $W_\phi^\theta$ and
$(W_\phi^\theta)^*$ (given by eq.(\ref{TWT}) and eq.(\ref{FIWF2}),
respectively) for $\theta=0$. Hence, they satisfy the properties
stated in Theorem \ref{Theorem12} (i) and (ii). They also satisfy
the intertwining relations
\begin{equation}
T_\phi \hat a= \hat A T_\phi \qquad , \qquad  T_\phi^* \hat A= \hat a T_\phi^* \, . \label{IRaA}
\end{equation}
The first relation is valid in $\S(\RE^n)$ and is just a restatement of eq.(\ref{Int}). The second relation is
valid in $\S(\RE^{2n})$ and was proved in \cite{Dias12}.

The next theorem is a consequence of these intertwining relations. It shows that the operators $\hat a$ and
$\hat A$ have equivalent spectral properties. The proof can also be found in \cite{Dias12}.

\begin{theorem}
Let $\hat a$ and $\hat A$ be the pseudo-differential operators given by eqs.(\ref{PDO},\ref{PSPDO}),
respectively. Let $\phi \in \S(\RE^n)$ be such that $||\phi||_{L^2(\RE^n)}=1$. Then

(i) The eigenvalues of $\hat a$ and $\hat A$ are the same.

(ii) If $\psi_\lambda$ is an eigenfunction of $\hat a$ then $\Psi_\lambda=T_\phi\psi_\lambda$ is an
eigenfunction of $\hat A$ (associated with the same eigenvalue).

(iii) Conversely, let $\Psi_\lambda$ be an eigenfunction of $\hat A$. If
$\psi_\lambda=T_\phi^*\Psi_\lambda\not=0$ then $\psi_\lambda$ is an eigenfunction of $\hat a$ (associated with
the same eigenvalue).

(iv) If $(\psi_\lambda)_\lambda$ is an orthonormal basis of eigenfunctions of $\hat a$ and $(\phi_\gamma\in
\S(\RE^n))_\gamma$ is an orthonormal basis of $L^2(\RE^n)$ then $(T_{\phi_\gamma}\psi_\lambda)_{\gamma,\lambda}$
is a complete set of eigenfunctions of $\hat A$ and forms an orthonormal basis of $L^2(\RE^{2n})$.

\label{Spectral}
\end{theorem}

The next theorem concerns the dynamical properties.

\begin{theorem}
Let $\hat a$ and $\hat A$ be given by eqs.(\ref{PDO},\ref{PSPDO}), respectively. Let $\phi \in \S(\RE^n)$
satisfy $||\phi||_{L^2(\RE^n)}=1$. Then:

(i) If $\psi(x,t)$ is the solution of the initial value problem
\begin{equation}
i \frac{\partial  \psi}{\partial t} = \hat a \psi \quad , \quad \psi(\cdot,0)= \psi_0(\cdot) \in \S(\RE^n)
\label{ivp1}
\end{equation}
then $\Psi=T_\phi \psi$ is the solution of
\begin{equation}
i \frac{\partial  \Psi}{\partial t} = \hat A \Psi \quad , \quad \Psi(\cdot,0)= \Psi_0(\cdot) \in \S(\RE^{2n})
\label{ivp2}
\end{equation}
where $\Psi_0=T_\phi \psi_0$.

(ii) Conversely, if $\Psi$ is the solution of the initial value problem (\ref{ivp2}) then $\psi=T^*_\phi \Psi$
is the solution of (\ref{ivp1}). \label{Dynamics}
\end{theorem}

The proof follows directly from the intertwining relations and the fact that the time derivative
operator commutes with both $T_\phi$ and $T^*_\phi$ \cite{Dias12}.\\

The results of the previous theorems can be extended to any other operator that is unitarily related with $\hat
A$. This is the case of the Bopp operators, formally defined by the mapping
\begin{equation}
a(x,\xi_x) \longmapsto \hat A_{B}=a(x+\frac{i}{2}\partial_p,p-\frac{i}{2}\partial_x). \label{FormalBopp}
\end{equation}
This map attributes to each phase space symbol $a \in \S'(\RE^{2n})$, an operator $\hat A_{B}$ acting on wave
functions $\Psi(x,p)$ with support on the phase space. The Bopp representation is closely related with the
deformation quantization of Bayen et. al \cite{Bayen1,Bayen2,Dias04-1,Dias04-2} and was used, in this context,
to prove some general spectral results for the stargenvalue equation \cite{Dias12,GoLu08}.

We now prove that the Bopp representation can be precisely defined by the mapping
\begin{equation}
a \longmapsto \hat A_{B}=\hat S \hat A \hat S^{-1} \label{BR}
\end{equation}
where $a\in \S'(\RE^{2n})$, $\hat A$ is given by eq.(\ref{PSPDO}) and $\hat S=\hat U(\theta_0)$ by eq.(\ref{U}).
Let us consider first the general result

\begin{proposition}[Metaplectic covariance property]
Let $\hat A \overset{\text{Weyl}}{\longleftrightarrow} A$ be a pseudo-differential operator with Weyl symbol
$A\in \S'(\RE^{4n})$, let $\hat S\in$ Mp$(4n)$ and let $S \in$ Sp$(4n)$ be its natural projection onto Sp$(4n)$.
Then
\begin{equation}
\hat A_S=\hat S \hat A \hat S^{-1} \overset{\text{Weyl}}{\longleftrightarrow} A_S=A \circ S^{-1} \, .
\label{MC}
\end{equation}
\label{prop14}
\end{proposition}

For a proof see \cite{Gosson11,Wong}. We then have:

\begin{theorem}
The operators $\hat A_{B}$ (given by eq.(\ref{BR})) are pseudo-differential operators with Weyl symbols
\begin{equation}
A_{B}(x,p;\xi_x,\xi_p)=a(x-\frac{\xi_p}{2},p+\frac{\xi_x}{2}) \, .\label{SBs}
\end{equation}
\end{theorem}

\begin{proof}
It follows from Proposition \ref{prop14} and the definition of $\hat A_{B}$ (\ref{BR}) that $\hat
A_{B}\overset{\text{Weyl}}{\longleftrightarrow} A_{B}=A\circ S^{-1}$ where
$A\overset{\text{Weyl}}{\longleftrightarrow} \hat A$ is given by (\ref{Sym}) and $S^{-1}$ is the inverse of the
projection of $\hat S$ onto Sp$(4n)$. Hence $S^{-1}=s(-\theta_0)$ (cf. Theorem \ref{CSym}) which can be easily
determined from (\ref{T1}) and (\ref{T2}). We have
\begin{equation}
S^{-1}:\RE^{4n} \longrightarrow \RE^{4n}, \quad \left\{
\begin{array}{l}
x \longmapsto x-\xi_p/2 \\
\xi_x \longmapsto \xi_x/2 +p\\
p \longmapsto p-\xi_x/2\\
\xi_p \longmapsto \xi_p/2 +x
\end{array} \right. .
\end{equation}
The result (\ref{SBs}) then follows from $A_{B}=A\circ S^{-1}$ by taking into account (\ref{Sym}).
\end{proof}

For the fundamental operators the mapping (\ref{BR}) yields:
$$
a(x,\xi_x)=x \Longrightarrow \hat A= x\cdot \Longrightarrow \hat A_B=x+\frac{i}{2}\partial_p
$$
$$
a(x,\xi_x)=\xi_x \Longrightarrow \hat A=-i\partial_x \Longrightarrow \hat A_B=p-\frac{i}{2}\partial_x
$$
and more generally, we have formally
$$
a(x,\xi_x) \longmapsto \hat A_B=a(x+\frac{i}{2}\partial_p,p-\frac{i}{2}\partial_x).
$$

Finally, it follows from $\hat A_{B}=\hat S \hat A \hat S^{-1}$ that the spectral and dynamical properties of
the Bopp and the {\it phase space} Schr\"odinger representations are equivalent. In view of Theorems
\ref{Spectral}, \ref{Dynamics} this equivalence can be extended to the {\it standard} Schr\"odinger
representation. Combining (\ref{IRaA}) with $\hat A_{B}=\hat S \hat A \hat S^{-1}$ we obtain the intertwining
relations
$$
\hat A_{B} \hat S T_\phi = \hat S T_\phi \hat a  \quad , \quad T_\phi^* \hat S^{-1} \hat A_{B}= \hat a T_\phi^*
\hat S^{-1}
$$
which can be re-written in terms of the windowed Wigner transform $W_\phi \equiv W_\phi^{\theta_0}$ and its
adjoint $ W_\phi^* \equiv \left( W^{\theta_0}_\phi\right)^*$ (given respectively by (\ref{TWT}) and
(\ref{FIWF},\ref{FIWF2}) for $\theta=\theta_0$)
$$
\hat A_{B} W_\phi=W_\phi \hat a \quad  , \quad  W_\phi^* \hat A_{B} = \hat a W_\phi^* \, .
$$
The Theorems \ref{Spectral} and \ref{Dynamics} are then valid {\it
ipsis verbis} for the pair of operators $\hat a$ and $\hat A_B$ if
we substitute the maps $T_\phi$ and $T^*_\phi$ by the maps
$W_\phi$ and $W^*_\phi$, respectively.

\vspace{0.7cm}

\noindent \textbf{Acknowledgements}. We thank the anonymous
referee for a detailed reading of the paper and for many useful
insights and remarks. Maurice de Gosson has been financed by the
Austrian Research Agency FWF (Projektnummer P20442-N13). Nuno
Costa Dias and Jo\~{a}o Nuno Prata have been supported by the
grant PTDC/MAT/099880/2008 of the Portuguese Science Foundation
(FCT).

\textbf{Author's addresses:}

\begin{itemize}
\item \textbf{Jo\~{a}o Nuno Prata }and\textbf{\ Nuno Costa Dias: }%
Departamento de Matem\'{a}tica. Universidade Lus\'{o}fona de
Humanidades e
Tecnologias. Av. Campo Grande, 376, 1749-024 Lisboa, Portugal and Grupo de F%
\'{\i}sica Matem\'{a}tica, Universidade de Lisboa, Av. Prof. Gama
Pinto 2, 1649-003 Lisboa, Portugal

\item \textbf{Maurice de Gosson:} Universit\"{a}t Wien,
Fakult\"{a}t f\"{u}r Mathematik--NuHAG, Nordbergstrasse 15, 1090
Vienna, Austria
\end{itemize}

\end{document}